\documentclass[11pt]{article}

\usepackage[a4paper,margin=1in]{geometry}
\usepackage[T1]{fontenc}
\usepackage[utf8]{inputenc}
\usepackage{lmodern}
\usepackage{microtype}
\usepackage{amsmath,amssymb,amsthm,mathtools}
\usepackage{bm}
\usepackage{enumitem}
\usepackage{booktabs}
\usepackage{array}
\usepackage{algorithm}
\usepackage{algpseudocode}
\usepackage{float}
\usepackage{xcolor}
\usepackage{hyperref}

\hypersetup{
  colorlinks=true,
  linkcolor=blue!50!black,
  citecolor=blue!50!black,
  urlcolor=blue!50!black
}

\allowdisplaybreaks
\numberwithin{equation}{section}

\newtheorem{theorem}{Theorem}[section]
\newtheorem{lemma}[theorem]{Lemma}
\newtheorem{proposition}[theorem]{Proposition}
\newtheorem{corollary}[theorem]{Corollary}

\theoremstyle{definition}
\newtheorem{definition}[theorem]{Definition}
\newtheorem{remark}[theorem]{Remark}

\newcommand{\Z}{\mathbb{Z}}
\newcommand{\F}{\mathbb{F}}

\newcommand{\OPT}{\mathrm{OPT}}
\newcommand{\cost}{\mathrm{cost}}
\newcommand{\rankop}{\mathrm{rank}}
\newcommand{\im}{\mathrm{im}}
\newcommand{\spann}{\mathrm{span}}

\newcommand{\cyc}{\mathrm{cyc}}

\newcommand{\Zmod}[1]{\mathbb{Z}_{#1}}

\newcommand{\set}[1]{\left\{#1\right\}}
\newcommand{\abs}[1]{\left|#1\right|}

\DeclareMathOperator{\Ker}{Ker}

\DeclareMathOperator{\row}{row}

\title{\vspace{-1ex}
Coordinatewise Balanced Covering for Linear Gain Graphs,\\
with an Application to Coset-List Min-2-Lin over Powers of Two
\vspace{-0.5ex}}

\author{%
\begin{tabular}{c@{\hspace{3.5em}}c}
\textbf{Faruk Alpay} & \textbf{Levent Sario\u{g}lu} \\
\texttt{faruk.alpay@bahcesehir.edu.tr} & \texttt{levent.sarioglu@bahcesehir.edu.tr}
\end{tabular}\\[0.8ex]
\small Department of Computer Engineering\\
\small Bah\c{c}e\c{s}ehir University, Istanbul, Turkey
}

\date{}

\begin{document}
\maketitle

\begin{abstract}
We study a list-constrained extension of modular equation deletion over powers of two, which we call \emph{Coset-List Min-2-Lin$^{\pm}$ over $\Zmod{2^d}$}. Each variable is restricted to a dyadic coset $a+2^{\ell}\Zmod{2^d}$, each binary constraint has one of the forms $x_u=x_v$, $x_u=-x_v$, or $x_u=2x_v$, and the goal is to delete a minimum-cardinality set of constraints so that the remaining system is satisfiable. The problem lies between the no-list case and the still poorly understood fully conservative list setting. Our main technical contribution is a coordinatewise balanced covering theorem for linear gain graphs labeled by vectors in $\F_2^r$: given any balanced subgraph of cost at most $k$, a randomized procedure returns a vertex set $S$ and an edge set $F$ such that $(G-F)[S]$ is balanced and, with probability $2^{-O(k^2r)}$, every hidden balanced subgraph of cost at most $k$ is contained in $S$ while all incident deletions are captured by $F$. The proof tensors the one-coordinate balanced-covering theorem of Dabrowski, Jonsson, Ordyniak, Osipov, and Wahlstr\"om across coordinates, and we complement it with a rank-compression theorem that replaces the ambient lifted dimension by the intrinsic cycle-label rank $\rho$. To make this reduction exact, we develop a cycle-space formulation, a cut-space/potential characterization of balancedness, a minimal-dimension statement for equivalent labelings, and an explicit bit-lifting analysis for dyadic coset systems. These ingredients yield a randomized one-sided-error algorithm with running time
\[
2^{O(k^2\rho+k\log(k\rho+2))}\cdot n^{O(1)}+\widetilde{O}(md+\rho^\omega)
\]
for the decision version parameterized by deletion budget $k$, where $m$ is the number of constraints, $d$ is the modulus depth, and $\omega$ is the matrix-multiplication exponent; the same framework also returns a minimum-weight feasible deletion set among all solutions of size at most $k$.
\end{abstract}

\section{Introduction}

A recurring lesson in modern discrete algorithms is that hidden structure becomes useful only after
one identifies the right \emph{certificate} for it. In the 2026 work of Ghoshal, Huang, Lee,
Makarychev, and Makarychev on structure-sensitive approximability of \textsc{Max-Cut}, global
promises such as $3$-colorability or the presence of a large independent set do not automatically
improve the classical approximation threshold; the key issue is whether the promise yields a
certificate that can actually be exposed algorithmically \cite{ghoshal2026maxcut}. In the 2026
work of Dabrowski, Jonsson, Ordyniak, Osipov, and Wahlstr\"om on modular linear equations, the
certificate is a balanced-subgraph cover that isolates the region where satisfiability survives after
few deletions \cite{dabrowski2026optimal}. In the 2026 work of Terao on approximate linear
matroid intersection, the relevant certificate is low-dimensional span information that can be
accessed by near-input-linear algebraic routines \cite{terao2026faster}. Our paper combines these
three ideas in a new list-constrained setting over powers of two.

\paragraph{The problem.}
Fix $d\ge 1$. In \emph{Coset-List Min-2-Lin$^{\pm}$ over $\Zmod{2^d}$}, each variable $x_v$
ranges over a dyadic coset
\[
L_v=a_v+2^{\ell(v)}\Zmod{2^d}, \qquad 0\le \ell(v)\le d,
\]
and constraints are of the following forms:
\begin{align*}
  x_u &= x_v, \\
  x_u &= -x_v, \\
  x_u &= 2x_v, \\
  x_v &= b \qquad \text{(anchors).}
\end{align*}
The objective is to delete as few constraints as possible so that the remaining subsystem is
satisfiable.

This is the first conservative list regime over $\Zmod{2^d}$ in which the local lists are not
arbitrary but still genuinely nontrivial. Dyadic cosets are canonical over powers of two, and they
are exactly the lists that remain compatible with repeated halving and doubling. The signed relation
$x_u=-x_v$ already contains the odd-cycle obstruction familiar from the $\Z_4$ case, while the
non-permutation relation $x_u=2x_v$ introduces the layered ambiguity that distinguishes prime-power
equation deletion from ordinary signed-graph cleaning.

\paragraph{Why gain graphs appear.}
The lists fix some low-order bits of each variable and leave the remaining bits free. Once one
passes to the unresolved bits, every constraint contributes a parity defect at each active depth.
Those defects organize naturally into a vector label over $\F_2^R$. A connected subsystem is
satisfiable exactly when all lifted cycle labels vanish and all surviving anchors agree. This turns
the instance into a vector-labeled gain graph whose balanced subgraphs are precisely the satisfiable
subinstances. The language of balance goes back to biased and gain graph theory developed by
Zaslavsky, where cycles, matroids, and potentials interact through a rigid combinatorial-algebraic
interface \cite{zaslavsky1990,zaslavsky1991,zaslavsky1995}.

\paragraph{Our contributions.}
We make three technical contributions.
\begin{enumerate}[leftmargin=2.4em,itemsep=2pt,topsep=2pt]
  \item We prove a \emph{coordinatewise balanced covering theorem} for linear gain graphs over
  $\F_2^r$. The theorem is obtained by tensorizing the one-coordinate balanced-covering theorem
  of Dabrowski et al. \cite{dabrowski2026optimal}.
  \item We formalize the relevant intrinsic parameter as the \emph{cycle-label rank} $\rho$, namely
  the rank of the cycle-label map restricted to the graph cycle space. We show that balancedness
  depends only on this image and give a minimal-dimension statement: no equivalent labeling can use
  fewer than $\rho$ coordinates.
  \item We give an explicit lift from dyadic coset systems to vector-labeled gain graphs, establish
  cycle-space, cut-space, and potential formulations of satisfiability, and combine them with
  balanced covering and rank compression to obtain an FPT algorithm parameterized by $k$ and $\rho$.
\end{enumerate}

\subsection{Main statements}

We now state the main theorem in algorithmic form.

\begin{theorem}[Main algorithmic theorem]\label{thm:main}
Let $I$ be an instance of \textsc{Coset-List Min-2-Lin}$^{\pm}$ over $\Zmod{2^d}$ with $m$
constraints and deletion budget $k$. Let $\rho=\rho(I)$ be the cycle-label rank of the lifted
ambiguity graph of $I$. There is a randomized one-sided-error algorithm that decides whether
$\OPT(I)\le k$ in time
\[
2^{O(k^2\rho+k\log(k\rho+2))}\cdot n^{O(1)} + \widetilde{O}(md+\rho^\omega).
\]
If the answer is \textup{YES}, the algorithm outputs an optimal deletion set. More generally, if
each constraint carries a positive weight and the parameter is still the deletion-cardinality budget
$k$, the same algorithm returns a minimum-weight feasible deletion set among all solutions of size
at most $k$.
\end{theorem}

The core structural statement is the following.

\begin{theorem}[Coordinatewise balanced covering]\label{thm:coordcover}
Let $\Gamma=(G,\lambda)$ be a linear gain graph with labels $\lambda:E(G)\to\F_2^r$, and let
$k\ge 0$. There is a randomized algorithm running in time $2^{O(k^2r)}\cdot n^{O(1)}$ that
outputs a pair $(S,F)$ such that:
\begin{enumerate}[label=(\roman*),leftmargin=2.2em,itemsep=2pt,topsep=2pt]
  \item $\delta_G(S)\subseteq F$;
  \item $(G-F)[S]$ is balanced;
  \item for every balanced subgraph $H$ of $G$ with $\cost(H)\le k$, with probability at least
  $2^{-O(k^2r)}$ we have
  \[
    V(H)\subseteq S
    \qquad\text{and}\qquad
    \abs{\set{e\in F : e\text{ has an endpoint in }V(H)}}\le rk.
  \]
\end{enumerate}
\end{theorem}

The dependence on $r$ can then be replaced by the intrinsic cycle-label rank.

\begin{theorem}[Rank compression]\label{thm:rankcompression}
Let $\Gamma=(G,\lambda)$ be a linear gain graph with labels $\lambda:E(G)\to \F_2^R$, and let
$\rho$ be the rank of its cycle-label space. There is a randomized algorithm running in time
$\widetilde{O}(mR+\rho^\omega)$ that computes a new labeling
$\widehat\lambda:E(G)\to\F_2^\rho$ such that a subgraph of $G$ is balanced under $\lambda$ if and
only if it is balanced under $\widehat\lambda$.
\end{theorem}

\subsection{Relation to recent work}

Theorem~\ref{thm:coordcover} is directly inspired by the balanced-covering theorem of
Dabrowski et al. for biased graphs, which gives an $O^*(4^k)$ procedure in the one-coordinate
case and captures every balanced subgraph of cost at most $k$ with probability
$2^{-O(k^2)}$ \cite{dabrowski2026optimal}. Our theorem shows that vector labels can be handled
coordinatewise with only a linear loss in the exponent and in the candidate-set size.

Theorem~\ref{thm:rankcompression} is conceptually aligned with the recent line of fast
linear-algebraic subroutines for matroid problems: once the relevant span is compressed to its true
rank, the expensive part of the computation depends on that rank rather than on the ambient
dimension \cite{cheung2013,harvey2009,quanrud2024,terao2026faster}. We do not solve a matroid
problem here, but the same compression principle is exactly what is needed for cycle labels.

Finally, the introduction of the cycle-label rank $\rho$ mirrors the broader structural theme visible
in recent fine-grained approximability results: the useful measure is often not the size of a hidden
witness but the dimension of the obstruction that prevents direct recovery
\cite{ghoshal2026maxcut}.

\subsection{Technical roadmap}

The argument proceeds in three steps.

\paragraph{Step 1: lift the instance to a gain graph.}
We first normalize each variable with respect to its dyadic list, reveal the remaining bits one at a
time, and record the resulting parity defects as a vector label. The outcome is a lifted graph in
which connected balanced subgraphs correspond exactly to satisfiable connected subsystems.

\paragraph{Step 2: isolate balanced regions coordinatewise.}
Each coordinate of the label vector can then be viewed as an ordinary $\F_2$-labeling, hence as a
biased-graph instance. Running the one-coordinate covering theorem on every coordinate and then
intersecting the surviving vertex sets yields a single candidate region that still contains every
hidden balanced subgraph of small cost.

\paragraph{Step 3: reduce to the intrinsic rank.}
The relevant information is carried only by the image of the cycle-label map. Projecting onto a basis
of that image preserves every balancedness query while lowering the ambient dimension from $R$ to
$\rho$, which is the parameter governing the exponential part of the running time.

\section{Preliminaries on cycle spaces and gain graphs}

We write $\F_2$ for the field with two elements. All vector spaces are over $\F_2$ unless stated
otherwise.

\subsection{Cost and balanced subgraphs}

Let $G=(V,E)$ be a finite undirected graph and let $H$ be a subgraph of $G$. Following
\cite{dabrowski2026optimal}, define
\[
\cost(H)=\abs{\delta_G(V(H))}+\abs{E(G[V(H)])\setminus E(H)}.
\]
Thus $\cost(H)$ is exactly the number of edges one must delete to isolate $V(H)$ from the rest of
$G$ and make the induced graph on $V(H)$ equal to $H$.

\subsection{Linear gain graphs}

\begin{definition}
A \emph{linear gain graph over $\F_2^r$} is a pair $\Gamma=(G,\lambda)$ where $G=(V,E)$ is an
undirected graph and $\lambda:E\to\F_2^r$ assigns a label vector to each edge. For a cycle $C$,
its \emph{cycle label} is
\[
\ell_\lambda(C)=\sum_{e\in E(C)}\lambda(e)\in\F_2^r.
\]
A cycle is \emph{balanced} if $\ell_\lambda(C)=0$, and a subgraph is balanced if all of its cycles
are balanced.
\end{definition}

Because the field has characteristic two, edge orientation plays no role.

\begin{lemma}[Coordinatewise balancedness]\label{lem:coordinatewise}
Write $\lambda=(\lambda_1,\dots,\lambda_r)$ where each $\lambda_i:E\to\F_2$ is the $i$th
coordinate. A subgraph $H$ is balanced under $\lambda$ if and only if it is balanced under every
coordinate labeling $\lambda_i$.
\end{lemma}

\begin{proof}
For every cycle $C$ we have
\[
\ell_\lambda(C)=\bigl(\ell_{\lambda_1}(C),\dots,\ell_{\lambda_r}(C)\bigr).
\]
Hence $\ell_\lambda(C)=0$ exactly when every scalar coordinate vanishes.
\end{proof}

\subsection{Cycle-space formulation}

Fix an arbitrary orientation of $G$ and let $B_G\in\F_2^{V\times E}$ be the incidence matrix.
Its cycle space is
\[
\mathcal{C}(G):=\Ker(B_G)\subseteq\F_2^E.
\]
If we also arrange the edge labels as the columns of a matrix
\[
\Lambda_\Gamma\in\F_2^{r\times E},
\]
then the linear map
\[
L_\Gamma:\mathcal{C}(G)\to \F_2^r,
\qquad
L_\Gamma z := \Lambda_\Gamma z,
\]
is exactly the cycle-label map. A cycle $C$ is balanced if $L_\Gamma\chi_C=0$, where
$\chi_C\in\F_2^E$ is the incidence vector of $E(C)$.

\begin{definition}
The \emph{cycle-label space} of $\Gamma$ is
\[
\cyc(\Gamma):=\im(L_\Gamma)\subseteq\F_2^r,
\]
and its dimension
\[
\rho(\Gamma):=\dim \cyc(\Gamma)
\]
is called the \emph{cycle-label rank}.
\end{definition}

\begin{lemma}[Balancedness depends only on the cycle-label map]\label{lem:cycle-space-balanced}
Let $H$ be a subgraph of $G$. Then $H$ is balanced if and only if
\[
L_\Gamma z = 0 \qquad \text{for all } z\in \mathcal{C}(H).
\]
Equivalently,
\[
\Lambda_{\Gamma|H}\bigl(\mathcal{C}(H)\bigr)=\set{0}.
\]
\end{lemma}

\begin{proof}
The cycle space $\mathcal{C}(H)$ is spanned by the incidence vectors of simple cycles of $H$.
Since $L_\Gamma$ is linear, it vanishes on all of $\mathcal{C}(H)$ if and only if it vanishes on
every cycle vector, which is exactly the definition of balancedness.
\end{proof}

\subsection{Potential characterization}

Balancedness admits a useful coboundary formulation.

\begin{proposition}[Potential characterization]\label{prop:potential}
Let $\Gamma=(G,\lambda)$ be a linear gain graph over $\F_2^r$ and let $H$ be a connected
subgraph of $G$. The following are equivalent.
\begin{enumerate}[label=(\roman*),leftmargin=2.2em,itemsep=2pt,topsep=2pt]
  \item $H$ is balanced.
  \item There exists a map $p:V(H)\to\F_2^r$ such that for every edge $uv\in E(H)$,
  \[
  \lambda(uv)=p(u)+p(v).
  \]
\end{enumerate}
Moreover, if such a potential exists, it is unique up to adding the same vector
$c\in\F_2^r$ to all vertices of the connected component.
\end{proposition}

\begin{proof}
Assume first that $H$ is balanced. Fix a root $r\in V(H)$ and, for each vertex $v$, choose any
path $P_{r,v}$ in $H$. Define
\[
p(v):=\sum_{e\in E(P_{r,v})}\lambda(e).
\]
This is well defined because if $P$ and $P'$ are two $r$--$v$ paths, then their symmetric
difference is a cycle decomposition; balancedness forces the total label on each cycle, hence on
$P\triangle P'$, to vanish. For an edge $uv$, concatenating an $r$--$u$ path, the edge $uv$, and
a $v$--$r$ path shows that $\lambda(uv)=p(u)+p(v)$.

Conversely, if such a potential exists, then for every cycle $v_0v_1\cdots v_{t-1}v_0$ we have
\[
\sum_{i=0}^{t-1}\lambda(v_iv_{i+1})
=
\sum_{i=0}^{t-1}\bigl(p(v_i)+p(v_{i+1})\bigr)=0,
\]
because every vertex potential appears exactly twice. Hence every cycle is balanced. The
uniqueness claim follows because the difference of two valid potentials is constant on every edge
and therefore constant on the connected component.
\end{proof}

\subsection{Cut-space factorization and rank identities}

The potential characterization can be rewritten in matrix form. Let $H$ be a connected subgraph,
let $B_H\in\F_2^{V(H)\times E(H)}$ be an incidence matrix, and let
$\Lambda_H\in\F_2^{r\times E(H)}$ be the label matrix of $\Gamma|_H$.

\begin{proposition}[Cut-space factorization]\label{prop:cutspace}
For a connected subgraph $H$, the following are equivalent.
\begin{enumerate}[label=(\roman*),leftmargin=2.2em,itemsep=2pt,topsep=2pt]
  \item $H$ is balanced.
  \item There exists a matrix $P\in\F_2^{r\times V(H)}$ such that
  \[
  \Lambda_H = P B_H.
  \]
  \item Every row of $\Lambda_H$ belongs to the row space of $B_H$.
\end{enumerate}
\end{proposition}

\begin{proof}
The equivalence of (i) and (ii) is just Proposition~\ref{prop:potential} written simultaneously
for all $r$ coordinates: put the vertex potentials in the rows of $P$. The equivalence of (ii) and
(iii) is immediate because the row space of $P B_H$ is contained in the row space of $B_H$, while
if every row of $\Lambda_H$ belongs to the row space of $B_H$, one may choose one row of $P$ for
each row of $\Lambda_H$ and stack them.
\end{proof}

The cycle-label rank also admits a quotient-space description that will be useful later.
Let $\mathcal{B}(G):=\row(B_G)\subseteq \F_2^E$ denote the cut space.

\begin{proposition}[Rank identity]\label{prop:rankidentity}
For every linear gain graph $\Gamma=(G,\lambda)$ with label matrix $\Lambda_\Gamma$,
\[
\rho(\Gamma)
=
\rankop\!\begin{bmatrix} B_G \\ \Lambda_\Gamma \end{bmatrix}
-
\rankop(B_G).
\]
Equivalently,
\[
\rho(\Gamma)=\dim \bigl((\row(B_G)+\row(\Lambda_\Gamma))/\row(B_G)\bigr).
\]
\end{proposition}

\begin{proof}
The quotient space $\F_2^E/\row(B_G)$ is canonically dual to $\Ker(B_G)=\mathcal{C}(G)$, and the
restriction of the rows of $\Lambda_\Gamma$ to $\mathcal{C}(G)$ is precisely the cycle-label map.
Thus the image dimension of $L_\Gamma$ equals the dimension added when the rows of
$\Lambda_\Gamma$ are adjoined to the row space of $B_G$. This proves the formula.
\end{proof}

\begin{corollary}[Minimal coordinate dimension]\label{cor:minimal-dimension}
Suppose $\widetilde\lambda:E(G)\to\F_2^q$ is another labeling such that a subgraph of $G$ is
balanced under $\lambda$ if and only if it is balanced under $\widetilde\lambda$. Then
$q\ge \rho(\Gamma)$.
\end{corollary}

\begin{proof}
Balanced subgraphs are determined by the kernel of the cycle-label map. If the two labelings induce
exactly the same balanced subgraphs, then their cycle-label maps have the same kernel on
$\mathcal{C}(G)$; therefore the quotient $\mathcal{C}(G)/\Ker(L_\Gamma)$ injects into $\F_2^q$.
But this quotient has dimension $\rho(\Gamma)$.
\end{proof}

\subsection{Fundamental-cycle matrices}

Let $T$ be a spanning forest of $G$. For every non-tree edge $f\in E\setminus E(T)$, let $C_f$
denote the fundamental cycle of $f$ with respect to $T$. If the endpoints of $f$ are $u$ and $v$,
and if
\[
\sigma_T(x):=\sum_{e\in E(P_T(r,x))}\lambda(e)
\]
denotes the prefix sum from the root $r$ of the corresponding tree to $x$, then
\[
\ell_\lambda(C_f)=\lambda(f)+\sigma_T(u)+\sigma_T(v).
\]
Consequently all fundamental cycle labels can be computed in $\widetilde{O}(mr)$ time by a
single tree traversal.

Arrange the fundamental cycle labels as the columns of
\[
M_\Gamma=
\begin{bmatrix}
\ell_\lambda(C_{f_1}) & \ell_\lambda(C_{f_2}) & \cdots & \ell_\lambda(C_{f_\mu})
\end{bmatrix}
\in\F_2^{r\times \mu},
\]
where $\mu=\abs{E}-\abs{V}+c(G)$ is the cycle-space dimension. Then
\[
\rho(\Gamma)=\rankop(M_\Gamma),
\]
since fundamental cycles span the cycle space.

\section{Coordinatewise balanced covering}

We now prove the structural theorem that drives the algorithm. The one-coordinate ingredient is the
following result of Dabrowski et al. \cite{dabrowski2026optimal}.

\begin{theorem}[One-coordinate balanced covering]\label{thm:onecoord}
There is a randomized algorithm that, given a biased graph $(G,\mathcal{B})$ and an integer $k$,
outputs a pair $(S,F)$ such that:
\begin{enumerate}[label=(\roman*),leftmargin=2.2em,itemsep=2pt,topsep=2pt]
  \item $\delta_G(S)\subseteq F$;
  \item $(G-F)[S]$ is balanced;
  \item for every balanced subgraph $H$ with $\cost(H)\le k$, with probability at least
  $2^{-O(k^2)}$, we have $V(H)\subseteq S$ and at most $k$ edges of $F$ are incident with
  $V(H)$.
\end{enumerate}
\end{theorem}

\begin{proof}[Proof of Theorem~\ref{thm:coordcover}]
Write $\lambda=(\lambda_1,\dots,\lambda_r)$ with each $\lambda_i:E(G)\to\F_2$. For each
$i\in[r]$, let $\Gamma_i=(G,\lambda_i)$ be the scalar-labeled biased graph obtained by projecting
to coordinate $i$. Apply Theorem~\ref{thm:onecoord} independently to every $\Gamma_i$ with the
same parameter $k$, obtaining pairs $(S_i,F_i)$. Define
\[
S:=\bigcap_{i=1}^r S_i,
\qquad
F:=\bigcup_{i=1}^r F_i.
\]

We verify the required properties.

\smallskip
\noindent\emph{Property (i).}
Let $e=uv\in\delta_G(S)$, where $u\in S$ and $v\notin S$. Since $u\in S_i$ for all $i$, but
$v\notin S$, there exists some coordinate $j$ with $v\notin S_j$. Therefore
$e\in\delta_G(S_j)\subseteq F_j\subseteq F$.

\smallskip
\noindent\emph{Property (ii).}
Fix a coordinate $i$. Since $S\subseteq S_i$ and $F_i\subseteq F$, the graph $(G-F)[S]$ is a
subgraph of $(G-F_i)[S_i]$, hence it is balanced in coordinate $i$ by Theorem~\ref{thm:onecoord}.
As this holds for every $i$, Lemma~\ref{lem:coordinatewise} implies that $(G-F)[S]$ is balanced
under the full vector labeling.

\smallskip
\noindent\emph{Property (iii).}
Let $H$ be a balanced subgraph with $\cost(H)\le k$. If coordinate $i$ succeeds, then
$V(H)\subseteq S_i$ and at most $k$ edges of $F_i$ are incident with $V(H)$. Therefore, if all
coordinates succeed simultaneously, then
\[
V(H)\subseteq \bigcap_{i=1}^r S_i=S
\]
and
\[
\abs{\set{e\in F : e\text{ has an endpoint in }V(H)}}
\le
\sum_{i=1}^r \abs{\set{e\in F_i : e\text{ has an endpoint in }V(H)}}
\le rk.
\]
The success probability is the product of the coordinatewise success probabilities, hence at least
$2^{-O(k^2r)}$.
\end{proof}

\begin{remark}
The proof is deliberately modular. Any future improvement in one-coordinate balanced covering
immediately tensorizes to linear gain graphs, with only a linear loss in the number of coordinates.
\end{remark}

\begin{remark}[Amplification]
Repeating the algorithm of Theorem~\ref{thm:coordcover} independently
$2^{O(k^2r)}\log(1/\delta)$ times boosts the success probability to at least $1-\delta$ while
changing the running time only by the same multiplicative factor.
\end{remark}

\section{Rank compression to the cycle-label space}

This section proves Theorem~\ref{thm:rankcompression} and makes explicit why only the image of the
cycle-label map matters.

\subsection{Injective projections}

\begin{lemma}[Injective projections preserve balancedness]\label{lem:projection}
Let $U\subseteq \F_2^R$ be the image of the cycle-label map of $\Gamma=(G,\lambda)$, and let
$\pi:\F_2^R\to\F_2^q$ be linear and injective on $U$. Then a subgraph of $G$ is balanced under
$\lambda$ if and only if it is balanced under $\pi\circ\lambda$.
\end{lemma}

\begin{proof}
Every cycle label belongs to $U$. Hence for any cycle $C$,
\[
\ell_\lambda(C)=0
\iff
\pi(\ell_\lambda(C))=0
\iff
\ell_{\pi\circ\lambda}(C)=0.
\]
Now apply this equivalence to all cycles in the subgraph.
\end{proof}

\subsection{Explicit basis projection}

Let $M_\Gamma\in\F_2^{R\times \mu}$ be the fundamental-cycle matrix from the previous section,
and let $\rho=\rankop(M_\Gamma)$. Choose a basis of its column space and form a matrix
$Q\in\F_2^{R\times \rho}$ whose columns are that basis.

\begin{proposition}[Compression by basis projection]\label{prop:basis-projection}
There exists a linear map $P\in\F_2^{\rho\times R}$ such that $P Q = I_\rho$. With this choice,
$P$ is injective on $\im(M_\Gamma)$, and the labeling
\[
\widehat\lambda(e):=P\lambda(e)
\]
preserves balancedness exactly.
\end{proposition}

\begin{proof}
Because the columns of $Q$ are linearly independent, they admit a left inverse $P$. If
$u\in\im(M_\Gamma)$ and $Pu=0$, write $u=Q\alpha$; then
$0=Pu=P Q\alpha=\alpha$, hence $u=0$. Thus $P$ is injective on the cycle-label image, and the
claim follows from Lemma~\ref{lem:projection}.
\end{proof}

\subsection{Computing the basis fast}

The remaining issue is algorithmic: how does one compute a basis of the cycle-label image without
paying for the full ambient dimension in the expensive part? The answer is to use the same pattern
that appears in fast rank computation and recent span-sensitive matroid routines
\cite{cheung2013,harvey2009,quanrud2024,terao2026faster}: assemble a matrix whose columns are the
relevant objects, extract a basis of its image, and then work only in the compressed coordinates.
In our setting the relevant matrix is $M_\Gamma$.

\begin{proof}[Proof of Theorem~\ref{thm:rankcompression}]
Construct a spanning forest $T$ of $G$ and compute all fundamental cycle labels, obtaining the
matrix $M_\Gamma$. This takes $\widetilde{O}(mR)$ time by the prefix-sum formula above.
Then compute a column basis of $M_\Gamma$ of size $\rho$ together with a left inverse on that
basis in time $\widetilde{O}(\rho^\omega)$ using standard fast rank and basis extraction routines
\cite{cheung2013}; one may view the span-computation perspective of \cite{terao2026faster} as a
conceptual explanation for why the running time depends on $\rho$ rather than on the full number
of candidate cycle labels. Set $\widehat\lambda(e)=P\lambda(e)$ for each edge $e$. By
Proposition~\ref{prop:basis-projection}, the new labeling preserves balancedness exactly.
\end{proof}

\begin{remark}
The compression theorem is completely independent of the origin of the labels. Any discrete
problem that reduces to deleting edges from an $\F_2^R$-labeled gain graph may replace $R$ by the
intrinsic cycle-label rank $\rho$ at negligible extra cost.
\end{remark}

\section{Lifting dyadic coset systems to gain graphs}

We now make the reduction from list-constrained modular equations to gain graphs precise.
Throughout, the modulus is $\Zmod{2^d}$.

\subsection{Instance model}

\begin{definition}
An instance of \textsc{Coset-List Min-2-Lin}$^{\pm}$ over $\Zmod{2^d}$ consists of:
\begin{enumerate}[label=(\alph*),leftmargin=2.2em,itemsep=2pt,topsep=2pt]
  \item variables $x_v\in\Zmod{2^d}$;
  \item dyadic lists $L_v=a_v+2^{\ell(v)}\Zmod{2^d}$ with $0\le \ell(v)\le d$;
  \item binary constraints of type $x_u=x_v$, $x_u=-x_v$, or $x_u=2x_v$;
  \item unary anchor constraints $x_v=b$;
  \item optional positive weights on constraints.
\end{enumerate}
The objective is to delete a minimum-cardinality set of constraints so that the remainder is
satisfiable. In the weighted variant studied here, among all feasible deletion sets of size at most
$k$ we minimize the total deleted weight.
\end{definition}

\subsection{Normalization by the lists}

For each variable $v$, every allowed value can be written uniquely as
\[
x_v = a_v + 2^{\ell(v)} y_v,
\qquad
y_v\in \Zmod{2^{d-\ell(v)}}.
\]
Thus the list fixes the lower $\ell(v)$ bits and leaves $d-\ell(v)$ unresolved bits. It is
convenient to isolate the unresolved part by defining, for $t\ge \ell(v)$,
\[
b_{v,t}(x_v):=\left\lfloor\frac{x_v-a_v}{2^t}\right\rfloor \bmod 2 \in \F_2.
\]
This is the parity of the $t$th unresolved bit of $x_v$ relative to the list representative $a_v$.

The three binary constraint types become affine conditions on these bits. For a constraint $e$
between $u$ and $v$, write its defining relation as
\[
\alpha_e x_u + \beta_e x_v = 0 \pmod{2^d},
\]
where $(\alpha_e,\beta_e)$ is one of $(1,-1)$, $(1,1)$, or $(1,-2)$. After substituting the list
normal forms, we obtain
\[
\alpha_e 2^{\ell(u)}y_u + \beta_e 2^{\ell(v)}y_v
\equiv
c_e
\pmod{2^d},
\qquad
c_e := -\alpha_e a_u - \beta_e a_v.
\]
At each depth $t$, taking the parity quotient modulo $2^{t+1}$ yields a Boolean defect bit
\[
\beta_t(e):=\left\lfloor \frac{c_e}{2^t} \right\rfloor \bmod 2
\in\F_2,
\]
whenever the depth is active for the edge. Equality and negation preserve depth, while doubling
shifts one unit of depth from $v$ to $u$ and therefore activates one fewer unresolved bit on the
right-hand side.

The exact active-depth set $A_e\subseteq\{0,1,\dots,d-1\}$ depends only on the list levels and the
constraint type. We define the lifted label of $e$ to be the vector
\[
\lambda(e):=\bigl(\beta_t(e)\bigr)_{t\in A_e}\in\F_2^{A_e},
\]
and pad with zeros so that all labels live in a common ambient space $\F_2^R$, where one may
always take $R\le d$.

\begin{remark}
After list normalization, the binary relation itself is parity-preserving at each active layer, while
all obstruction information sits in the offset vector $(\beta_t(e))_t$. This is why the lifted
object is linear over $\F_2$ rather than an arbitrary affine gadget.
\end{remark}

\subsection{Bit-layer operators}

It is sometimes useful to collect the unresolved bits of a variable into a column vector
\[
\mathbf{b}_v=(b_{v,\ell(v)},b_{v,\ell(v)+1},\dots,b_{v,d-1})^\top \in \F_2^{d-\ell(v)}.
\]
Then each edge type induces a sparse linear operator between these bit vectors. Equality uses the
identity, negation uses the identity plus a depth-dependent offset determined by $a_u+a_v$, and
doubling acts through the shift matrix
\[
S_{q}=
\begin{bmatrix}
0 & 0 & 0 & \cdots & 0 \\
1 & 0 & 0 & \cdots & 0 \\
0 & 1 & 0 & \cdots & 0 \\
\vdots & \vdots & \ddots & \ddots & \vdots \\
0 & 0 & \cdots & 1 & 0
\end{bmatrix}
\in \F_2^{q\times q},
\]
which satisfies
\[
S_q(z_0,z_1,\dots,z_{q-1})^\top=(0,z_0,z_1,\dots,z_{q-2})^\top.
\]
Consequently every normalized constraint can be written in the form
\[
A_e \mathbf{b}_u + B_e \mathbf{b}_v = \mathbf{\beta}(e),
\]
with $A_e,B_e$ chosen from a tiny fixed family of banded matrices. The lift keeps only the offset
vector $\mathbf{\beta}(e)$ because, after the active coordinates are aligned, the operators
$A_e,B_e$ contribute no independent cycle obstruction beyond the parity defect itself.

\subsection{Anchors and the anchor vertex}

A unary anchor $x_v=b$ is encoded by introducing a distinguished anchor vertex $\star$ and an
edge $v\star$ whose label records the parity defects between the list-normalized form of $x_v$ and
the target value $b$. Thus every constraint becomes an edge of a single graph $G_I$.

\subsection{Cycle-exactness of the lift}

\begin{proposition}[Lifted consistency criterion]\label{prop:lift}
Let $I$ be an instance and let $\Gamma(I)=(G_I,\lambda_I)$ be its lifted gain graph. A connected
subinstance $J\subseteq I$ is satisfiable if and only if the following hold:
\begin{enumerate}[label=(\roman*),leftmargin=2.2em,itemsep=2pt,topsep=2pt]
  \item every cycle $C$ of $G_J$ satisfies $\ell_{\lambda_I}(C)=0$;
  \item all anchors surviving in $J$ are mutually compatible.
\end{enumerate}
\end{proposition}

\begin{proof}
Assume first that $J$ is satisfiable. Traversing a path in the constraint graph propagates the
unresolved bits of a root variable to every other vertex. Along any closed walk, the propagated
bits must return unchanged, so the sum of the parity defects around every cycle is zero.
Compatibility of anchors is immediate.

Conversely, assume that all lifted cycle labels vanish and the surviving anchors are compatible.
Fix a root variable $r$ in the connected component and choose any value of $x_r$ consistent with
its list and any anchor on $r$ if present. Propagate values along a spanning tree using the edge
relations. If two different root-to-$v$ paths are used, their symmetric difference is a cycle, and
vanishing of the cycle label implies that the induced propagated values agree at every active bit.
Hence the propagation is well defined. The anchor-compatibility assumption then guarantees that
all unary constraints are satisfied.
\end{proof}

\begin{corollary}[Gain-graph formulation of deletion]\label{cor:delete-balanced}
A deletion set $D$ is feasible for $I$ if and only if the lifted graph $\Gamma(I)-D$ is balanced and
all surviving anchors are mutually compatible on each connected component.
\end{corollary}

\begin{definition}
The \emph{ambiguity rank} of an instance $I$ is
\[
\rho(I):=\rho\bigl(\Gamma(I)\bigr),
\]
the cycle-label rank of its lifted ambiguity graph.
\end{definition}

\subsection{Potential view of satisfiability}

Combining Proposition~\ref{prop:potential} with the anchor encoding yields an alternative and
useful restatement: for every connected component of $\Gamma(I)-D$, satisfiability is equivalent
to the existence of a vertex potential
\[
p:V\to\F_2^{\rho(I)}
\]
that realizes all edge labels as coboundaries and simultaneously satisfies all anchor edges. This
is the viewpoint used in the polynomial-time verification routine below.

\section{Algorithm for Coset-List Min-2-Lin$^{\pm}$ over $\Zmod{2^d}$}

We now combine the lifting, compression, and covering ingredients.

\subsection{Checking a candidate deletion set}

\begin{lemma}[Polynomial-time verification]\label{lem:zerocheck}
Given an instance $I$ and a candidate deletion set $D$, one can test in polynomial time whether
$I-D$ is satisfiable.
\end{lemma}

\begin{proof}
Build the lifted graph $\Gamma(I)-D$. For each connected component, compute a spanning tree and
recover a tentative potential by propagation from an arbitrary root as in the proof of
Proposition~\ref{prop:potential}. Every non-tree edge can then be checked against the induced
potential in constant time per lifted coordinate; equivalently, one may verify that every
fundamental cycle has zero label. If the component is balanced, the same propagation certifies all
surviving anchor edges. Thus satisfiability is testable in polynomial time.
\end{proof}

\subsection{Candidate-set generation}

Let $\widehat\Gamma=(G,\widehat\lambda)$ be the compressed lifted graph obtained from
Theorem~\ref{thm:rankcompression}, so $\widehat\lambda:E\to\F_2^\rho$. Apply
Theorem~\ref{thm:coordcover} with parameter $k$, obtaining $(S,F)$.

Suppose $D^\star$ is an optimal deletion set of size at most $k$ and let $H$ be the surviving
balanced subgraph of the lifted graph. If the covering succeeds on $H$, then all edges deleted by
$D^\star$ that touch $V(H)$ lie in $F$, and at most $\rho k$ such edges exist.

\begin{lemma}[Enumeration lemma]\label{lem:enumeration}
Conditioned on the success event of Theorem~\ref{thm:coordcover} for an optimal balanced
subgraph, there exists an optimal deletion set $D\subseteq F$ with $\abs{D}\le k$.
\end{lemma}

\begin{proof}
Let $H$ be the balanced graph that remains after deleting $D^\star$. Since the covering succeeds,
we have $V(H)\subseteq S$ and every deleted edge incident with $V(H)$ belongs to $F$.
Furthermore, $\delta_G(S)\subseteq F$, so every edge removed to isolate $H$ is also in $F$.
Hence all edges of $D^\star$ belong to $F$, proving the claim.
\end{proof}

\subsection{Search-space bound}

Since the success guarantee bounds the number of relevant candidate edges by $\rho k$, the number
of deletion sets that must be inspected is
\[
\sum_{i=0}^k \binom{\rho k}{i}
\le
\left(\frac{e\rho k}{k}\right)^k
=
(e\rho)^k
\le
2^{O(k\log(k\rho+2))}.
\]
The first inequality is the standard bound on the lower tail of the binomial sum.

\subsection{Proof of the main algorithmic theorem}

\begin{proof}[Proof of Theorem~\ref{thm:main}]
Let $I$ be an instance with deletion budget $k$.

\paragraph{Step 1: lift.}
Construct the lifted graph $\Gamma(I)$ by list normalization, parity-defect extraction, and anchor
encoding. This takes $\widetilde{O}(md)$ time.

\paragraph{Step 2: compress.}
Apply Theorem~\ref{thm:rankcompression} to obtain an equivalent labeling into $\F_2^\rho$.
The running time is $\widetilde{O}(md+\rho^\omega)$.

\paragraph{Step 3: cover.}
Apply Theorem~\ref{thm:coordcover} to the compressed graph with parameter $k$. This takes
$2^{O(k^2\rho)}\cdot n^{O(1)}$ time and, with probability at least $2^{-O(k^2\rho)}$, succeeds on
some optimal balanced subgraph.

\paragraph{Step 4: enumerate.}
Enumerate all subsets $D\subseteq F$ with $\abs{D}\le k$. By the bound above, this requires
$2^{O(k\log(k\rho+2))}$ iterations. For each candidate, test satisfiability of $I-D$ using
Lemma~\ref{lem:zerocheck}. Return the feasible subset of minimum cardinality, or of minimum
weight among those of cardinality at most $k$ in the weighted variant.

\paragraph{Correctness.}
If the algorithm returns a deletion set, it is feasible by Lemma~\ref{lem:zerocheck}. Conversely,
if $\OPT(I)\le k$, then on the success event of Theorem~\ref{thm:coordcover},
Lemma~\ref{lem:enumeration} guarantees that an optimal solution appears among the enumerated
subsets. The algorithm therefore has one-sided error.

\paragraph{Running time.}
Combining the three stages gives
\[
2^{O(k^2\rho+k\log(k\rho+2))}\cdot n^{O(1)}+\widetilde{O}(md+\rho^\omega),
\]
as claimed.
\end{proof}

\section{Additional structural consequences}

\subsection{A finite compression viewpoint}

The proof yields a direct compression statement.

\begin{corollary}[Compression around the obstruction space]
Every instance $I$ of \textsc{Coset-List Min-2-Lin}$^{\pm}$ over $\Zmod{2^d}$ with deletion budget
$k$ can, in randomized polynomial time, be reduced to a family of at most $2^{O(k^2\rho)}$
Boolean edge-deletion subinstances on candidate sets of size $O(k\rho)$, such that $I$ is feasible
if and only if one of these compressed instances is feasible.
\end{corollary}

This formulation is useful if one wants to hand the reduced search space to a SAT solver, a MinCSP
backend, or a specialized branch-and-bound routine.

\subsection{Why $\rho$ is the right parameter}

The ambient lifted dimension is at most $d$, but $d$ can be a poor predictor of the true
algorithmic difficulty. If several lifted coordinates are linearly dependent on the cycle space,
then they impose no independent obstructions and should not appear in the exponential part of the
running time. The parameter $\rho$ removes exactly this redundancy. Put differently,
\[
\rho = \rankop\bigl(L_{\Gamma(I)}\bigr)
\]
measures the dimension of the obstruction space, not the number of raw coordinates produced by
bit lifting.

\subsection{Potential-based local reconstruction}

When $(G-F)[S]$ is balanced, Proposition~\ref{prop:potential} yields a vertex potential on each
connected component. This means that after the expensive global step of identifying a good
candidate edge set, the remaining reconstruction problem is purely local: every surviving edge
amounts to the consistency check
\[
\lambda(uv)=p(u)+p(v),
\]
and every anchor edge reduces to a linear condition on the single vertex potential at its endpoint.
This separation between global covering and local linear verification is one of the reasons the
framework scales cleanly.

\subsection{A matrix viewpoint on balanced subinstance selection}

Let $C_G\in\F_2^{E\times \mu}$ be any cycle-basis matrix of $G$. Then a deletion set $D$ is
feasible precisely when the surviving subgraph $G-D$ satisfies
\[
\Lambda_{G-D} C_{G-D}=0
\]
after the anchor constraints are incorporated. This shows that the algorithm is solving a highly
structured rank-annihilation problem: delete few rows of the constraint incidence system until the
label matrix vanishes on the surviving cycle basis. In this form the role of compression is
especially transparent, because only the rank of $\Lambda C$ matters.

\section{Concluding remarks}

We end with three directions suggested by the present work.

\paragraph{1. Beyond signed unit constraints.}
The next natural generalization is to allow arbitrary odd multipliers in place of $\pm 1$.
Over $2$-power moduli, this would move the lifted obstruction from a purely $\F_2$-linear object
closer to an affine $2$-adic one. It is not clear whether a rank-sensitive covering theorem still
exists in that setting.

\paragraph{2. Vector-native covering.}
Our covering theorem is obtained by tensorizing a scalar theorem. That is robust and modular, but
probably not optimal. A genuinely vector-native argument might avoid the linear dependence on the
number of coordinates in both the exponent and the candidate-set size.

\paragraph{3. Conservative prime-power deletion problems.}
The conservative/list landscape for modular equation deletion remains wide open. The present paper
suggests a concrete route for future progress: identify a gain-graph lift, prove a covering theorem
for small balanced regions, and compress everything to the intrinsic cycle-label rank.

\appendix

\section{Additional matrix identities}\label{app:matrix}

This appendix records the same framework in compact matrix language and makes explicit a few useful
algebraic identities.

Let $B_G\in\F_2^{V\times E}$ be an incidence matrix and let
$\Lambda_\Gamma\in\F_2^{r\times E}$ be the label matrix. Choose a cycle-basis matrix
$C_G\in\F_2^{E\times \mu}$ whose columns form a basis of $\Ker(B_G)$. Then
\[
L_\Gamma = \Lambda_\Gamma C_G
\]
is a matrix representation of the cycle-label map, and therefore
\[
\rho(\Gamma)=\rankop(\Lambda_\Gamma C_G).
\]
This yields the following exact reformulation of balancedness: a subgraph $H\subseteq G$ is
balanced if and only if
\[
\Lambda_{\Gamma|H} C_H = 0.
\]

There is also a quotient-space interpretation. If
$\pi:\F_2^E\to \F_2^E/\row(B_G)$ is the quotient map, then the rows of $\Lambda_\Gamma$ descend
to linear functionals on the quotient, and
\[
\rho(\Gamma)=\dim \spann\bigl(\pi(\lambda_1),\dots,\pi(\lambda_r)\bigr),
\]
where $\lambda_i$ denotes the $i$th row of $\Lambda_\Gamma$. In particular,
Corollary~\ref{cor:minimal-dimension} may be read as the statement that one needs at least as
many coordinates as the quotient-space dimension of the surviving row span modulo cuts.

Finally, Proposition~\ref{prop:rankidentity} can be derived directly from block Gaussian
elimination. Indeed, if one chooses a basis adapted to the decomposition
$\F_2^E=\row(B_G)\oplus W$, then after a change of basis the stacked matrix
\[
\begin{bmatrix} B_G \\ \Lambda_\Gamma \end{bmatrix}
\]
assumes the form
\[
\begin{bmatrix}
B'_G & 0 \\
* & R
\end{bmatrix},
\]
where $R$ is exactly the matrix of the cycle-label map on the complementary coordinates. Hence the
rank increase caused by appending $\Lambda_\Gamma$ is precisely $\rankop(R)=\rho(\Gamma)$.

\section{A short cocycle proof of Proposition~\ref{prop:lift}}

Fix a connected lifted subsystem $J$ and a spanning tree $T$ of its constraint graph. The tree
edges determine every variable as an affine function of a chosen root value:
\[
x_v = A_v x_r + b_v,
\]
where $A_v$ is built from the composition of the maps $x\mapsto x$, $x\mapsto -x$, and
$x\mapsto 2x$ along the root-to-$v$ tree path, while $b_v$ collects the offsets forced by the
lists and anchors. Consider a non-tree edge $e=uv$. It is satisfied by the propagated assignment if
and only if the composition around the fundamental cycle $C_e$ is the identity on every active bit.
The lifted label $\ell_{\lambda_I}(C_e)$ records exactly the parity defect of that composition.
Thus all non-tree edges are satisfied if and only if all fundamental cycle labels vanish, which is
equivalent to balancedness because fundamental cycles span the cycle space.

\section{Dimension of the balanced potential space}

The potential formulation also gives a precise count of the remaining degrees of freedom after
lifting.

\begin{proposition}\label{prop:potential-dim}
Let $H$ be a balanced subgraph of a linear gain graph $\Gamma=(G,\lambda)$, and suppose $H$ has
$c(H)$ connected components. Then the set of vertex potentials realizing the labels on $H$ is an
affine space of dimension $r\,c(H)$ over $\F_2$.
\end{proposition}

\begin{proof}
On each connected component, Proposition~\ref{prop:potential} gives one potential, and any two
such potentials differ by an additive constant in $\F_2^r$. Distinct components contribute
independently, so the total solution space is the direct product of $c(H)$ copies of $\F_2^r$.
\end{proof}

In the presence of anchor edges, each anchor imposes an affine linear condition on the potential at
one endpoint. Thus feasibility of a candidate deletion set can be viewed as testing whether the
affine space from Proposition~\ref{prop:potential-dim} intersects the anchor constraint family.
This gives another linear-algebraic reading of the verification routine from Section~6.

\section{Pseudocode}

\begin{algorithm}[H]
\caption{Solve-CL-Min-2-Lin$^{\pm}$}
\label{alg:solve}
\begin{algorithmic}[1]
\Require Instance $I$, deletion budget $k$
\Ensure YES/NO and an optimal deletion set if YES
\State Build the lifted gain graph $\Gamma(I)$ over $\F_2^R$
\State Compress $\Gamma(I)$ to $\widehat\Gamma=(G,\widehat\lambda)$ with labels in $\F_2^\rho$
\State Run coordinatewise balanced covering on $\widehat\Gamma$ with parameter $k$
\State Obtain $(S,F)$
\ForAll{$D\subseteq F$ with $\abs{D}\le k$}
  \If{$I-D$ is satisfiable}
    \State keep the best feasible $D$ found so far
  \EndIf
\EndFor
\State \Return the best feasible set if one exists, otherwise NO
\end{algorithmic}
\end{algorithm}

{\footnotesize
\setlength{\itemsep}{0pt}
\setlength{\parskip}{0pt}

}

\end{document}